\documentclass{svmult}
\usepackage{mathptmx}
\usepackage{amssymb}
\usepackage{helvet}
\usepackage{courier}
\usepackage{makeidx}
\usepackage{graphicx}
\usepackage{multicol}
\usepackage[bottom]{footmisc}
\usepackage{footmisc}
\usepackage{url}
\usepackage{enumerate}
\usepackage{color}
\usepackage{amsmath}
\usepackage{bm}

\definecolor{mygreen}{cmyk}{0.7,     0,      0.9,      0}

\def\be{\begin{equation}}
\def\ee{\end{equation}}

\newcommand{\bfv}{\boldsymbol{v}}
\newcommand{\bfx}{{\boldsymbol{x}}}
\newcommand{\bfu}{\boldsymbol{u}}

\newcommand{\ba}{\boldsymbol{a}}
\newcommand{\Vv}{\boldsymbol{V}}
\newcommand{\Xx}{\boldsymbol{X}}
\newtheorem{thm}{Theorem}[section]

\newtheorem{lem}[thm]{Lemma}
\newtheorem{cor}[thm]{Corollary}

\newtheorem{rem}[thm]{Remark}

\allowdisplaybreaks

\begin{document}

\title*{The Maxwell-Stefan Diffusion Limit of a Hard-Sphere Kinetic Model for Mixtures }
\author{Benjamin Anwasia}
\institute{Benjamin Anwasia \at Centro de Matem\'atica da Universidade do Minho, Campus de Gualtar,4710-057 Braga, Portugal \email{id6226@alunos.uminho.pt}
}
%
%
\maketitle

\abstract*{Each chapter should be preceded by an abstract (10--15 lines long) that summarizes the content. The abstract will appear \textit{online} at \url{www.SpringerLink.com} and be available with unrestricted access. This allows unregistered users to read the abstract as a teaser for the complete chapter. As a general rule the abstracts will not appear in the printed version of your book unless it is the style of your particular book or that of the series to which your book belongs.
Please use the 'starred' version of the new Springer \texttt{abstract} command for typesetting the text of the online abstracts (cf. source file of this chapter template \texttt{abstract}) and include them with the source files of your manuscript. Use the plain \texttt{abstract} command if the abstract is also to appear in the printed version of the book.}

\abstract{We study a kinetic model for non-reactive mixtures of monatomic gases with hard-sphere cross-sections under isothermal condition. 
By considering a diffusive scaling of the kinetic model and using the method of moments, we formally obtain from the continuity and momentum balance equations of the species, in the limit as the scaling parameter goes to zero, the Maxwell-Stefan diffusion equations, with an explicit expression for the diffusion coefficients.}
\section{Introduction}
\label{sec:1}
The study of diffusion phenomena is very important due to its varied uses in many fields, such as engineering, physics, biology, chemistry, etc. The most classical constitutive law used to describe diffusive transport was given by Fick in \cite{fic1, fick2}. Fick postulated that flux goes from higher concentration regions  to lower concentration regions  with a magnitude that is proportional to the concentration gradient. This proportionality relation between flux and concentration gradient postulated by Fick is predominantly used to model diffusion. In many situations, it provides accurate description of diffusive transport, while in others, it seems to be too simplistic (as in the case of multispecies/multicomponent mixtures) and hence cannot be used to describe diffusion in these cases. 

The limitations of using Fick's  constitutive relation to describe diffusion in multicomponent gaseous mixtures  and some other situations have been shown experimentally, see for example \cite{Dun-Tor-62,  KW-CES-97}. In the case of multicomponent gaseous mixtures, three distinct diffusion phenomena, referred to as osmotic diffusion, uphill or reverse diffusion and diffusion barrier were observed. These three types of diffusion phenomena cannot be described by the Fickian approach. More precisely, osmotic diffusion corresponds to a situation of diffusion without a gradient. Uphill or reverse diffusion refers to a situation in which flux goes from lower concentration regions to higher concentration regions. Diffusion barrier is a situation  of diffusion  in which flux is zero.

Due to the  shortcomings of the Fick's constitutive relation in describing diffusion in multicomponent mixtures, a more general constitutive relation  known as the Maxwell-Stefan (MS) equations  \cite{max1866, ste1871} is often used instead. 
The MS equations rely on the fact that the driving force of the species in a multicomponent mixture is in local equilibrium with the total inter-species drag/friction force. 

In spite of the importance of the MS equations in describing diffusion in multi-species mixtures, its mathematical study is relatively new. In particular, \cite{BGP-NA-17,bou-gre-sal-15, HS-MMAS-17} dealt with the formal derivation of diffusion models of MS type from a diffusive scaling of Boltzmann type equations for non-reactive multi-species mixtures, under isothermal condition (i.e. uniform in space and constant in time mixture temperature). The existence and uniqueness issues, as well as the long-time behaviour of solutions of the MS diffusion systems have been considered in \cite{Bothe11, bou-gre-sal-12, Jungel13}. The numerical study of  MS equations has been considered in  \cite{McB14}. The derivation of a non-isothermal diffusion model of MS type from a kinetic model for non-reactive mixtures has been considered in \cite{HS-NA-17}. The derivation of an isothermal reaction diffusion model of MS type from the simple reacting sphere (SRS) kinetic model has been considered in  \cite{AGS-19}. More precisely, from a scaling of the SRS kinetic model which corresponds to a situation where the dominant role in the evolution of the species is played by mechanical interaction, while chemical reactions are assumed to be slow enough to allow the surroundings to continually compensate for the difference in heat between the reactants and products. Here an explicit expression for the MS diffusion coefficients was obtained. The derivation of  non-isothermal  diffusion model of MS type from a kinetic model for a reactive mixture of polyatomic gases with a continuous structure of internal energy  has been considered in \cite{ABSS-20}. More precisely, from a scaling where mechanical collisions are dominant while chemical reactions are slow. Here the MS diffusion coefficients  are not explicit because a general cross-sections was considered.

The goal of this work is to derive an isothermal diffusion model of MS type,  with explicit expressions for the MS diffusion coefficients, as a hydrodynamic limit of a kinetic model for non-reactive mixtures, with hard-sphere cross-sections. Also, the computations are given in detail in order to show some particular aspects and technicalities involved in the derivation. We emphasize that the main difference between this work and \cite{BGP-NA-17,bou-gre-sal-15, HS-MMAS-17} is the cross-sections considered. More precisely, \cite{BGP-NA-17} and \cite{HS-MMAS-17} considered  general and analytic cross-sections, respectively, and as a result, the MS diffusion coefficients obtained were not explicit. In \cite{bou-gre-sal-15}, Maxwellian molecules were considered. Here explicit expressions for the diffusion coefficients were obtained, but they are different from those obtained in the present work.

The rest of this work is organized as follows. In section 2, we introduce the kinetic model for a non-reactive mixture with hard-sphere cross-sections. In section 3,  we introduce the scaled kinetic equation, its properties and the assumptions that will be needed in our analysis. 
In section 4, we present the MS diffusion limit of the scaled kinetic equations. More preciesely, we obtain the species continuity equations and the momentum balance equations for the species from  the scaled kinetic equations. Finally, we present the asymptotic analysis of the species continuity and momentum balance equations  towards the diffusion model of MS type. Section 6 is dedicated to our conclusions.
\section{Kinetic Model}
\label{sec:2}
The starting point of our analysis is a system of Boltzmann-type equations that describe the evolution of a non-reactive mixture of $N$ monatomic inert gases, $A_i$ with $i=1,2, \dots ,N$. Particles in the mixture undergo elastic collisions of hard-sphere type. These elastic collisions occur between particles of the same species and between particles of different species. More precisely, in the absence of external forces, let  $f_i\!:=\! f_i(t,\bfx, \bfv_i)\!\geq\!0$ be the unknown probability distribution function, representing the density of particles of species $A_i$ which at time $t$ are located at position ${\bfx}$ and have velocity ${\bfv}_i$, we will study the following Cauchy problem:
\begin{equation}
\begin{aligned}
\frac{\partial f_i}{\partial t}\!+\!{\bfv}_i\!\cdot\! \frac{\partial f_i}{\partial {\bfx}} &= J_i+\sum_{\substack{s=1\\ s\neq i}}^{N}J_{is},\quad (t,\bfx, \bfv_i)\in\mathbb{R}_+ \times \mathbb{R}^3 \times \mathbb{R}^3,\, i\!=\!1,2,3,4,
\\[-0.2em]
f_i(0 , \bfx , \bfv_i) &= (f_i^{\rm in}) (\bfx , \bfv_i), \quad  (\bfx,\bfv_i)\in \mathbb{R}^3 \times \mathbb{R}^3,
\end{aligned}
\label{eq:Scaled_KE*}
\end{equation}
where 
\begin{equation}
J_i=\sigma_{ii}^2 \int _{{\mathbb{R}}^3}\int _{{\mathbb{S}^2_+}}
             \left [ {f}_i{'}{f'_{i_*}}-f_if_{i_*} \right ]\left \langle \epsilon ,{\bfv}_i-{\bfv}_{i_*} \right \rangle d\epsilon \,d{\bfv}_{i_*}
\label{eq:J_i}
\end{equation}
is the mono-species collision operator and it represents collisions between particles of the same species,
\vspace{-0.2cm}
\begin{equation}
J_{is}=\sigma^2 _{is} \int_{{\mathbb{R}}^3} \int _{{\mathbb{S}^2_+}}
              \left [ {f}_i{'}{f'_s}-f_if_s \right] \left \langle \epsilon ,{\bfv}_i-{\bfv}_s \right\rangle d\epsilon\, d{\bfv}_s
\label{eq:J_is}
\end{equation}
is the bi-species collision operator and it represents  collisions between particles of different species.

\noindent In equations \eqref{eq:J_i} and \eqref{eq:J_is} for the mono-species and bi-species collision operators,
${f}_i' \!=\! f(t,{\bfx},{\bfv}'_i)$, 
$ {f}_{i*}' \!=\! f(t,{\bfx},{\bfv}'_{i*})$, 
${f}_s' \!=\! f(t,{\bfx},{\bfv}'_s)$, 
$\epsilon$ is a unit vector directed along the line joining the centre of the two spheres at the moment of impact (i.e. $\epsilon \in \mathbb{S}_+^2 \!=\! \big\{ \overline\epsilon\in \mathbb{R}^3: \; \| \overline\epsilon\| = 1, \; \langle\overline\epsilon,{\bf{v}}_i - {\bf{v}}_s\rangle > 0\big\}$),  $\left\langle \cdot , \cdot \right\rangle$ represents the inner product in $\mathbb{R}^3$ and
$\|\cdot\|$ is  the norm induced by this inner product.
Furthermore, ${\bfv}_i, {\bfv}_{i*}, {\bfv}_s $ are the pre-collisional velocities, 
${\bfv}'_i, {{\bfv}'_{i_*}}, {\bfv}'_s $ are the post-collisional velocities, $\sigma_{ii*}^2$ and  $\sigma_{is}^2$ are the mono-species and the bi-species collision cross-sections defined respectively as
\begin{equation*}
\sigma _{ii}^2=d_i^2, \qquad\sigma _{is}^2=\frac{1}{4}(d_i+d_s)^2,  \qquad i,s=1,2,3,4, \quad i\neq s,
\end{equation*}
where $d_i$ and $d_s$ are the diameters of particles of species $A_i$ and $A_s$.

 \smallskip
 
Since the collisions are elastic, both momentum and kinetic energy are conserved. Therefore, the conservation laws of linear momentum and kinetic energy for bi-species collisions are given by:
\begin{align}
&m_i{\bfv}_i\!+\!m_s{\bfv}_s \!=\! m_i{\bfv}_i'+m_s{\bfv}_s',
\label{eq:ECOM}
\\[-0.1em]
&\frac{1}{2} m_i({\bfv}_i)^2\!+\frac{1}{2}\!m_s({\bfv}_s)^2\! =\! \frac{1}{2}m_i({\bfv}_i')^2\!+\!\frac{1}{2}m_s ({\bfv}_s')^2 ,
             \label{eq:ECOKE}
             \end{align}
respectively. The post-collisional velocities can be written in terms of the pre-collisional velocities as
\begin{equation}
{\bfv}_i' = {\bfv}_i-2\frac{\mu_{is}}{m_i} \epsilon \left \langle \epsilon ,{\bfv}_i - {\bfv}_s \right \rangle \quad   \mbox{and} \quad {\bfv}_s'= {\bfv}_s+2\frac{\mu_{is}}{m_s}\epsilon \left \langle \epsilon ,{\bfv}_i-{\bfv}_s \right \rangle,
\label{eq:EPCV}
\end{equation}
where $\displaystyle \mu_{is}=\frac{m_im_s}{m_i+m_s}$ is the reduced mass.
\begin{rem}
 For same species collisions, the indices $i$ and $i*$ are used to distinguish their velocities. Also, the conservation of linear momentum and kinetic energy for mono-species collisions are respectively, given by:
\begin{align}
m_i{\bfv}_i\!+\!m_i{\bfv}_{i*} &= m_i{\bfv}_i'+m_i{\bfv}_{i*}',
\label{eq:ECOM_mono}
\\[-0.1em]
             m_i({\bfv}_i)^2\!+\!m_i({\bfv}_{i*})^2& = m_i({\bfv}_i')^2\!+\!m_i ({\bfv}_{i*}')^2.
\label{eq:ECOKE_mono}
\end{align}
Furthermore, the post-collisional velocities can be written in terms of the pre-collisional velocities as given below:
\begin{equation}
{\bfv}_i' = {\bfv}_{i}- \epsilon \left \langle \epsilon ,{\bfv}_i - {\bfv}_{i*} \right \rangle \quad   \mbox{and} \quad {\bfv}_{i*}'= {\bfv}_{i*}+\epsilon \left \langle \epsilon ,{\bfv}_i-{\bfv}_{i*} \right \rangle.
\end{equation}
\end{rem}
\begin{rem}
 The kinetic equations given in the first row of \eqref{eq:Scaled_KE*}  with $J_i$ and $J_{is}$  defined in  \eqref{eq:J_i} and \eqref{eq:J_is}, respectively, can be obtained from the simple reacting sphere kinetic model studied in \cite{PS2017} when the chemical reaction is turned off.
\end{rem}
\section{Scaled Kinetic Equations, Properties and Assumptions}
In this section, we first present the scaled version of the kinetic equations given in the first row of \eqref{eq:Scaled_KE*}. Then we present some of  its properties and  the assumptions that will be used to obtain the desired diffusion model.
\subsection{Scaled Kinetic Equations}
Introducing a reference length and time scales together with a reference temperature, one can define the dimensionless time, space, velocities (obtained using the speed of sound in a monatomic gas), cross-sections and number densities. Thus the kinetic equations given in the first row of \eqref{eq:Scaled_KE*},  where $J_i$ and $J_{is}$ are as defined in \eqref{eq:J_i} and \eqref{eq:J_is}, respectively, can be written in the following dimensionaless/scaled form 
\begin{equation}
\begin{aligned}
\alpha\frac{\partial f_i^\alpha}{\partial t}\!+\!{\bfv}_i\!\cdot\! \frac{\partial f_i^\alpha}{\partial {\bfx}} \!&=\frac{1}{\alpha}\underbrace{\sigma_{ii}^2 \int _{{\mathbb{R}}^3}\int _{{\mathbb{S}^2_+}}
             \left [ {f}_i^{\alpha '}{f^{\alpha '}_{i_*}}-f_i^\alpha f_{i_*}^\alpha \right ]\left \langle \epsilon ,{\bfv}_i-{\bfv}_{i_*} \right \rangle d\epsilon \,d{\bfv}_{i_*}}_{J_i^{\alpha}} 
             \\[-0.2em]             
              &+\frac{1}{\alpha} \sum_{\substack{s=1\\ s\neq i}}^{N}
              \underbrace{\sigma^2 _{is} \int_{{\mathbb{R}}^3} \int _{{\mathbb{S}^2_+}}
              \left [ {f}_i^{\alpha'}{f^{\alpha '}_s}-f^\alpha_i f^\alpha_s \right] \left \langle \epsilon ,{\bfv}_i-{\bfv}_s \right\rangle d\epsilon\, d{\bfv}_s}_{J_{is}^{\alpha}},
\label{eq:scaled_KE}
\end{aligned}
\end{equation}
for $(t,\bfx, \bfv_i)\in\mathbb{R}_+ \times \mathbb{R}^3 \times \mathbb{R}^3,\, i=1,2,3,4,$ where  $\alpha$ which is such that $0 < \alpha \ll 1$ is the formal scaling parameter representing the Knudsen number, $f_i^\alpha,\, f_i^{\alpha'},\, f_{i*}^{\alpha},\, f_{i*}^{\alpha'},\, f_s^{\alpha}$ and $f_s^{\alpha'}$ are the scaled distribution function, $J_{i}^{\alpha}$ is the scaled mono-species collision operator and $J_{is}^{\alpha}$ is the scaled bi-species collision operator. We have also assumed that the Mach number is of the same order of magnitude as the Knudsen number. 

\smallskip

In what follows, the evolution domain of the mixture is represented by an open bounded domain $\Omega \subset \mathbb{R}^3$, with regular boundary and we will consider the following  Cauchy problem:
\begin{equation}
\begin{aligned}
\alpha\frac{\partial f_i^\alpha}{\partial t}\!+\!{\bfv}_i\!\cdot\! \frac{\partial f_i^\alpha}{\partial {\bfx}} \!\!&=\!\!\frac{1}{\alpha}\bigg(\! J_i^{\alpha}\!+\!\sum_{\substack{s=1\\ s\neq i}}^{N} J_{is}^{\alpha}\!\bigg),\quad (t,\bfx, \bfv_i)\in\mathbb{R}_+ \times \mathbb{R}^3 \times \mathbb{R}^3,\, i\!=\!1,2,3,4,
\\
 f_i^\alpha(0 , \bfx , \bfv) &= (f_i^{\rm in})^\alpha (\bfx , \bfv), \quad  (\bfx,\bfv)\in \mathbb{R}^3 \times \mathbb{R}^3.
\end{aligned} 
\label{eq:Scaled_KE}
\end{equation}
\subsection{Properties of the Collision Operators}
Here we will present some properties of the mono-species and bi-species collision operators that will be needed in our analysis.
\begin{lem}
\label{lem:weakform_MS}
Let $\varphi({\bf{v}}_i)$ be a sufficiently smooth test function. Then, the weak form of the mono-species collision operator $J_i^{\alpha}$ is given by
\begin{multline}
\int_{\mathbb{R}^3}\!J_i^{\alpha}\varphi({\bfv}_i)\,d{\bfv}_i 
      \!=\!  \frac{1}{4} \sigma_{ii}^2 \int_{\mathbb{R}^3}\!\int _{{\mathbb{R}}^3}\!\int 
      \left [ {f}_i^{\alpha'}{f}_{i*}^{\alpha'}-f_i^\alpha f_{i*}^\alpha \right ]\left \langle \epsilon ,{\bfv}_i-{\bfv}_{i*} \right \rangle 
      \\
      \times\left [ \varphi({\bfv}_i)\!+\!\varphi({\bfv}_{i*})\!-\!\varphi({\bfv}'_i)\!-\!\varphi({\bfv}'_{i*}) \right ] 
     d\epsilon \,d{\bfv}_{i*} d{\bfv}_i.
\label{eq:weakform_MSECO}
\end{multline}
\end{lem}
\begin{lem}
\label{lem:weakform_BS}
let $\varphi({\bf{v}}_i)$ be a sufficiently smooth test function. Then, the weak form of the bi-species collision operator  $J_{is}^{\alpha}$ is given by 
\begin{equation}
\int_{\mathbb{R}^3}\!\!J_{is}^{\alpha}\varphi({\bfv}_i)d{\bfv}_i
       \!=\! \sigma^2_{is}\!\!\int_{\mathbb{R}^3}\!\! \int_{\mathbb{R}^3}\!\! \int_{\mathbb{S}_+^2}\!\!\! \left [ \varphi({\bfv}_i')\!-\!\varphi ({\bfv}_i) \right ] f_i^\alpha\! f_s^\alpha \left \langle \epsilon ,{\bfv}_i\!-\!{\bfv}_s \right \rangle d\epsilon d{\bfv}_s d{\bfv}_i,
\label{eq:weakform_Qisdvi}
\end{equation}
for each $i,s=1,2,3,4,$ with $i\not=s$.
\end{lem}
The proofs of lemmas \ref{lem:weakform_MS} and \ref{lem:weakform_BS} follow  from  standard arguments, see for example \cite{bou-gre-sal-15, Carlo}.
\begin{cor}
   \label{cor:monospecies_MMKE_con}
The  mono-species  collision  operator $J_i^\alpha$ is such that, for $i=1,2,3,4$,
\begin{equation}
\int_{\mathbb{R}^3}J_i^\alpha
\begin{pmatrix}1\\ 
m_i{\bf{v}}_i\\ 
\frac{1}{2}m_i({\bf{v}}_i)^2
\end{pmatrix}d{\bf{v}}_i =0.
\label{eq:Mconservation_MMKE}
\end{equation}
\end{cor}
\begin{proof}
\smartqed
The proof follows from lemma \ref{lem:weakform_MS} by setting $\varphi ({\bf{v}}_i)=1$, $\varphi ({\bf{v}}_i)=m_i {\bf{v}}_i$ and using the mono-species conservation of momentum \eqref{eq:ECOM_mono}, $\varphi ({\bf{v}}_i)=\frac{1}{2}m_i ({\bf{v}}_i)^2$ and using the mono-species conservation of energy \eqref{eq:ECOKE_mono}, respectively.
\qed
\end{proof}
\begin{cor}
\label{cor:Bconservation_MMKE}
The  bi-species  collision  operator $J_{is}^\alpha$ is such that, for  $i,s=1,2,3,4,$ with $i\not=s$,
\begin{equation}
\begin{aligned}
&\int_{\mathbb{R}^3}J^\alpha_{is}\,d{\bf{v}}_i=0,
\\
 &\int_{\mathbb{R}^3}J^\alpha_{is}
              \begin{pmatrix}
             m_i{\bf{v}}_i\\ 
             \frac{1}{2}m_i({\bf{v}}_i)^2
             \end{pmatrix}\,d{\bf{v}}_i+ \int_{\mathbb{R}^3}J^\alpha_{si}
             \begin{pmatrix} 
             m_s{\bf{v}}_s\\ 
             \frac{1}{2}m_s({\bf{v}}_s)^2
             \end{pmatrix}\,d{\bf{v}}_s=0.
\end{aligned}
\label{eq:Bconservation_MMKE}
\end{equation}
\end{cor}
\begin{proof}
\smartqed
The proof follows from lemma \ref{lem:weakform_BS} by setting $\varphi ({\bf{v}}_i)=1$, $\varphi ({\bf{v}}_i)=m_i {\bf{v}}_i$ and using the bi-species conservation of momentum  \eqref{eq:ECOM}, $\varphi ({\bf{v}}_i)=\frac{1}{2}m_i ({\bf{v}}_i)^2$ and using the bi-species conservation of energy \eqref{eq:ECOKE}, respectively.
\qed
\end{proof}
\subsection{Assumptions}In order to derive the target equations (that is a purely diffusion model where the diffusion process is governed by the MS equations), we will assume that
\begin{enumerate}[\it(a)]
\item The temperature  of the mixture $T$ is uniform in space and constant in time.
\item The bulk velocity  of the  mixture ${\bf{u}}^\alpha$ is small and goes to zero as $\alpha \to 0$.
\item The initial conditions are local Maxwellians centered at the average velocity of the species ${\bf{u}}_i^\alpha$. More precisely the distribution functions at time $t=0$ are of the form
\begin{equation*}
              f_i^{\alpha(in)} \! ({\bfx},{\bfv}_i)\! =\! c_i^{\alpha(in)} \! ({\bfx}) \!\!
              \left ( \! \frac{m_i}{2\pi k_BT} \! \right )^{\!\!\frac{3}{2}}\!\! 
             \exp \! \left[-\frac{m_i \Big( \! {\bfv}_i \!-\! \alpha{\bfu}^{\alpha(in)}_i({\bfx}) \! \Big)^{\!2}}{2k_BT}\right] \! ,
              \; {\bfx}\in \Omega, \; {\bfv}_i\in \mathbb{R}^3,
\label{eq:local_max}
\end{equation*}
where $k_B$ is the Boltzmann's constant, $c_i^\alpha$ is the concentration of the species which is such that  $c_i^{\alpha(in)}:  \,\Omega\rightarrow {\mathbb{R}}_+$ 
and ${\bf{u}}_i^{\alpha (in)} \! : \Omega\rightarrow {\mathbb{R}^3}$, for $i=1, \ldots ,4$.
\item The evolution of the system leaves the distribution functions in the local Maxwellian state. More precisely, the distribution functions at time $t>0$ is of the form
\begin{equation}
f_i^{\alpha}(t,{\bfx},{\bfv}_i) =
       c_i^{\alpha}(t,{\bfx})
       \left ( \! \frac{m_i}{2\pi k_BT} \! \right )^{\!\!\frac{3}{2}} \!
       \exp \! \left[-\frac{m_i \Big( \! {\bfv}_i \!-\! \alpha{\bfu}^{\alpha}_i(t,{\bfx}) \! \Big)^2}{2k_BT}\right], {\bfx}\in \Omega, \; {\bfv}_i\in \mathbb{R}^3
\label{eq:local_max1}
\end{equation}
where  $c_i^{\alpha}:{\mathbb{R}}_+ \times \Omega\rightarrow {\mathbb{R}}_+$ 
and $ {\bf{u}}_i^{\alpha}:{\mathbb{R}}_+ \times \Omega\rightarrow {\mathbb{R}^3}$,
for $i=1,\ldots4$.
\end{enumerate}
Assumption (a)  is the isothermal condition and it allows us to neglect effects due to temperature gradient. Assumption (b) allows us to neglect convective effects. Assumptions (c) and (d) represents a physical situation in which the system evolves not far away from the local Maxwellian Equilibrium.
\begin{lem}
\label{lem:property}
As a consequence of \eqref{eq:local_max1}, we have that
 \begin{multline}
     {f}_i^{\alpha'}{f_s}^{\alpha'}-f^\alpha_if^\alpha_s=M_i^\alpha M_s^\alpha\big( \alpha {\ba}_{is}\cdot({\bfv}_i'-{\bfv}_i)+\alpha^2{\ba}_{is}\cdot({\bfv}_i'-{\bfv}_i)  ({\ba}_i\cdot{\bfv}_i)
     \\
     +\alpha^2 {\ba}_{is}\cdot ({\bfv}_i'-{\bfv}_i) ({\ba}_s\cdot{\bfv}_s)+\frac{\alpha^2}{2}\big({\ba}_{is}\cdot ({\bfv}_i'-{\bfv}_i)\big)^2+O(\alpha^3) \big), 
     \label{eq:f_prime_minus_f}
     \end{multline}
where
\begin{equation}
\begin{aligned}
&{\ba}_i=\frac{m_i{\bfu}_i}{k_BT}, \quad {\ba}_s=\frac{m_s{\bfu}_s}{k_BT}, \quad {\ba}_{is}=\frac{ m_i({\bfu}_i^\alpha-{\bfu}_s^\alpha)}{k_BT},
\\[-0.1em]
&M_i^\alpha\!\!=\!\!c_i^{\alpha}\!\bigg(\! \frac{m_i}{2\pi k_BT} \!\bigg)^\frac{3}{2}\!\!\!\exp\!\bigg(\!\!-\!\frac{m_i({\bfv}_i)^2}{2k_BT}\!\!\bigg), \,\, M_s^\alpha\!\!=\!\!c_s^{\alpha}\!\bigg(\! \frac{m_s}{2\pi k_BT} \!\bigg)^\frac{3}{2}\!\!\exp\!\bigg(\!\!-\!\frac{m_s({\bfv}_s)^2}{2k_BT}\!\bigg).
\label{eq:definitions}
\end{aligned}
\end{equation}
\end{lem}
\begin{proof}
\smartqed
Using  equation \eqref{eq:local_max1}, we can write
\begin{align}
\nonumber{f}_i^{\alpha'}\!{f}_{s}^{\alpha'}\!&=\!c_i^{\alpha}\!\left ( \! \frac{m_i}{2\pi k_BT} \! \right )^{\!\!\frac{3}{2}} \!\!\exp \! \left[\!-\!\frac{m_i \Big( \! {\bfv}_i' \!-\! \alpha{\bfu}^{\alpha}_i \! \Big)^2}{2k_BT}\!\right] \!c_s^{\alpha}\!\left ( \! \frac{m_s}{2\pi k_BT} \! \right )^{\!\!\frac{3}{2}} \!\!\exp \! \left[\!-\!\frac{m_s \Big( \! {\bfv}_s' \!-\! \alpha{\bfu}^{\alpha}_s \! \Big)^2}{2k_BT}\!\right]
\\
&\!=\!c_i^{\alpha}c_s^{\alpha}  \frac{(m_i m_s)^{\frac{3}{2}}}{(2\pi k_BT)^3}\exp \!\! \left[\!-\Bigg(\!\frac{{m_i \Big( \! {\bfv}_i' \!-\! \alpha{\bfu}^{\alpha}_i \! \Big)^2}\!+\!{m_s \Big( \! {\bfv}_s' \!-\! \alpha{\bfu}^{\alpha}_s \! \Big)^2}}{2k_B T}\!\Bigg)\!\right].
\label{eq:ff_prime}
\end{align}
Observe that
\begin{align*}
m_i\left ( {\bfv}_i'-\alpha{\bfu}_i^\alpha \right )^2+m_s\left ( {\bfv}_s'-\alpha{\bfu}_s^\alpha \right )^2 &=m_i\left \{ {\bfv}_i'({\bfv}_i'-\alpha{\bfu}_i^\alpha)-\alpha{\bfu}_i^\alpha({\bfv}_i'-\alpha{\bfu}_i^\alpha) \right \}
\\
&+m_s\left \{ {\bfv}_s'({\bfv}_s'-\alpha{\bfu}_s^\alpha)-\alpha{\bfu}_s^\alpha({\bfv}_s'-\alpha{\bfu}_s^\alpha) \right \}
\\
&=m_i\left \{ ({\bfv}_i')^2-2\alpha{\bfu}_i^\alpha\cdot{\bfv}_i'+(\alpha{\bfu}_i^\alpha)^2 \right \}
\\
&+m_s\left \{ ({\bfv}_s')^2-2\alpha{\bfu}_s^\alpha\cdot{\bfv}_s'+(\alpha{\bfu}_s^\alpha)^2 \right \}
\\
& =m_i ({\bfv}_i')^2-2\alpha m_i{\bfu}_i^\alpha\cdot{\bfv}_i'+m_i(\alpha{\bfu}_i^\alpha)^2 
\\
&+m_s ({\bfv}_s')^2-2\alpha m_s{\bfu}_s^\alpha\cdot{\bfv}_s'+ m_s(\alpha{\bfu}_s^\alpha)^2
\\
&=m_i ({\bfv}_i')^2+m_s ({\bfv}_s')^2 +m_i(\alpha{\bfu}_i^\alpha)^2
\\
&+m_s(\alpha{\bfu}_s^\alpha)^2-2\alpha m_i{\bfu}_i^\alpha\cdot{\bfv}_i' -2\alpha m_j{\bfu}_s^\alpha\cdot{\bfv}_s'.
\end{align*}
 Using the bi-species conservation of kinetic energy \eqref{eq:ECOKE}, we obtain
 \begin{align*}
m_i\left ( {\bfv}_i'-\alpha{\bfu}_i^\alpha \right )^2+ m_s\left ( {\bfv}_s'-\alpha{\bfu}_s^\alpha \right )^2 &=m_i {\bfv}_i^2+m_s {\bfv}_s^2+m_i(\alpha{\bfu}_i^\alpha)^2
\\
&+m_s(\alpha{\bfu}_s^\alpha)^2 -2\alpha m_i{\bfu}_i^\alpha\cdot{\bfv}_i' -2\alpha m_s{\bfu}_s^\alpha\cdot{\bfv}_s'
\\
&=m_i\left ( {\bfv}_i-\alpha{\bfu}_i^\alpha \right )^2+m_s\left ( {\bfv}_s-\alpha{\bfu}_s^\alpha \right )^2
\\
&+2\alpha\big( m_i{\bfu}_i^\alpha\!\cdot\!{\bfv}_i\!-\!m_i{\bfu}_i^\alpha\!\cdot\!{\bfv}_i'\!+\! m_s{\bfu}_s^\alpha\!\cdot\!{\bfv}_s \!-\! m_s{\bfu}_s^\alpha\!\cdot\!{\bfv}_s'\big)
\\
 &=m_i\left ( {\bfv}_i-\alpha{\bfu}_i^\alpha \right )^2+m_s\left ( {\bfv}_s-\alpha{\bfu}_s^\alpha \right )^2
 \\
 &-2\alpha m_i{\bfu}_i^\alpha \cdot( {\bfv}_i'-{\bfv}_i)-2\alpha m_s{\bfu}_s^\alpha\cdot({\bfv}_s'-{\bfv}_s).
 \end{align*}
  Using the bi-species conservation of momentum \eqref{eq:ECOM} rewritten as
     \begin{equation*}
      m_i({\bfv}_i'-{\bfv}_i)=-m_s({\bfv}_s'-{\bfv}_s),
     \end{equation*}
we obtain
 \begin{align}
\nonumber m_i\left ( {\bfv}_i'-\alpha{\bfu}_i^\alpha \right )^2+ m_s\left ( {\bfv}_s'-\alpha{\bfu}_s^\alpha \right )^2 &= m_i\left ( {\bfv}_i-\alpha{\bfu}_i^\alpha \right )^2+ m_s\left ( {\bfv}_s-\alpha{\bfu}_s^\alpha \right )^2
\\
\nonumber&-2\alpha m_i {\bfu}_i^\alpha \cdot( {\bfv}_i'-{\bfv}_i)+2\alpha m_i{\bfu}_s^\alpha \cdot({\bfv}_i'-{\bfv}_i)
\\
\nonumber& =m_i\left ( {\bfv}_i-\alpha{\bfu}_i^\alpha \right )^2+ m_s\left ( {\bfv}_s-\alpha{\bfu}_s^\alpha \right )^2
\\
&-2\alpha m_i({\bfv}_i'-{\bfv}_i)\cdot  ({\bfu}_i^\alpha-{\bfu}_s^\alpha).
\label{eq: quadratic_terms_expansion}
 \end{align}
Substituting \eqref{eq: quadratic_terms_expansion} into \eqref{eq:ff_prime}, we obtain
\begin{multline*}
{f}_i^{\alpha'}{f}_{s}^{\alpha'}\!=c_i^{\alpha}c_s^{\alpha}  \frac{(m_i m_s)^{\frac{3}{2}}}{(2\pi k_BT)^3}
\\
 \times \exp \bigg[\!-\Bigg(\!\frac{m_i\left ( {\bfv}_i-\alpha{\bfu}_i^\alpha \right )^2+ m_s\left ( {\bfv}_s-\alpha{\bfu}_s^\alpha \right )^2-2\alpha m_i({\bfv}_i'-{\bfv}_i)\cdot  ({\bfu}_i^\alpha-{\bfu}_s^\alpha)}{2k_B T}\!\Bigg)\!\bigg].
\end{multline*}
Using the distribution function \eqref{eq:local_max1}, we obtain
\begin{equation}
 {f}_i^{\alpha'}{f}_{s}^{\alpha'}= {f}_i^{\alpha}{f}_{s}^{\alpha}\exp \bigg[\frac{\alpha m_i({\bfv}_i'-{\bfv}_i)\cdot ({\bfu}_i^\alpha-{\bfu}_s^\alpha)}{k_BT} \bigg]. 
 \label{eq:ff_prime_final}
\end{equation}
Now, using \eqref{eq:ff_prime_final}, we can write
\begin{align}
\nonumber{f}_i^{\alpha'}\!\!{f_s}^{\alpha'}\!-\!f^\alpha_i\!\!f^\alpha_s&\!= {f}_i^{\alpha}{f}_{s}^{\alpha}\exp \bigg[\frac{\alpha m_i({\bfv}_i'-{\bfv}_i)\cdot ({\bfu}_i^\alpha-{\bfu}_s^\alpha)}{k_BT} \bigg] -{f}_i^{\alpha}{f}_{s}^{\alpha}
\\
&=  {f}_i^{\alpha}{f}_{s}^{\alpha}\bigg(\exp \bigg[\frac{\alpha m_i({\bfv}_i'-{\bfv}_i)\cdot ({\bfu}_i^\alpha-{\bfu}_s^\alpha)}{k_BT} \bigg] -1\bigg).
\label{eq:fprime_minus_f}
\end{align}
Taylor expanding the exponential term in \eqref{eq:fprime_minus_f} above with respect to $\alpha$ gives
 \begin{multline}
     \exp \bigg[{\frac{\alpha m_i({\bfv}_i'-{\bfv}_i)\cdot ({\bfu}_i^\alpha-{\bfu}_s^\alpha)}{k_B T}}\bigg]=1+\frac{\alpha m_i({\bfv}_i'-{\bfv}_i)\cdot ({\bfu}_i^\alpha-{\bfu}_s^\alpha)}{k_B T}
     \\
     +\frac{1}{2}\left ( \frac{\alpha m_i({\bfv}_i'-{\bfv}_i)\cdot ({\bfu}_i^\alpha-{\bfu}_s^\alpha)}{k_B T} \right )^2+O(\alpha^3) .
     \label{eq:Exp_expansion1}
     \end{multline}
     Substituting \eqref{eq:Exp_expansion1} into \eqref{eq:fprime_minus_f} gives
     \begin{align}
   \nonumber {f}_i^{\alpha'}\!\!&{f_s}^{\alpha'}\!-\!f^\alpha_i\!\!f^\alpha_s
   \\
   \nonumber&\!\!={f}_i^{\alpha}{f}_{s}^{\alpha}\bigg(1+\frac{\alpha m_i({\bfv}_i'-{\bfv}_i)\cdot ({\bfu}_i^\alpha-{\bfu}_s^\alpha)}{k_BT}+\frac{1}{2}\left ( \frac{\alpha m_i({\bfv}_i'-{\bfv}_i)\cdot ({\bfu}_i^\alpha-{\bfu}_s^\alpha)}{k_B T} \right )^2
   \\
   \nonumber&+O(\alpha^3)  -1\bigg)
    \\
    &\!\!=\!{f}_i^{\alpha}\!{f}_{s}^{\alpha}\!\bigg(\!\frac{\alpha m_i({\bfv}_i'\!-\!{\bfv}_i)\!\cdot\! ({\bfu}_i^\alpha\!-\!{\bfu}_s^\alpha)}{k_B T}\!+\!\frac{1}{2}\!\left (\! \frac{\alpha m_i({\bfv}_i'\!-\!{\bfv}_i)\!\cdot\! ({\bfu}_i^\alpha\!-\!{\bfu}_s^\alpha)}{k_B T} \!\right )^{\!\!2}\!+\!O(\alpha^3)\! \bigg).
    \label{eq:ff}
     \end{align}
Observe that 
\begin{align*}
    \exp\bigg[-\frac{m_i({\bfv}_i-\alpha{\bfu}^{\alpha}_i)^2}{2k_BT}\bigg]&=\exp \bigg[-\frac{m_i\left \{ ({\bfv}_i)^2-2\alpha{\bfu}^{\alpha}_i\cdot{\bfv}_i+(\alpha{\bfu}_i^\alpha)^2 \right \}}{2k_BT}\bigg]
    \\
    &=\exp \bigg[-\frac{m_i({\bfv}_i)^2}{2k_BT}\bigg] \,\exp\bigg[ \frac{m_i\alpha{\bfu}_i^\alpha\cdot{\bfv}_i}{k_BT} \bigg] \,\exp\bigg[-\frac{m_i(\alpha{\bfu}_i^\alpha)^2}{2k_BT}\bigg].
    \end{align*}
Taylor expanding the last two terms on the right hand side of the previous equation with respect to $\alpha$ gives
\begin{align*}
&\exp\bigg[ \frac{m_i\alpha{\bfu}_i^\alpha\cdot{\bfv}_i}{k_BT} \bigg]= 1+\alpha\frac{m_i{\bfu}_i^\alpha\cdot{\bfv}_i}{k_BT}+ O(\alpha^2),
\\
&\exp\bigg[-\frac{m_i(\alpha{\bfu}_i^\alpha)^2}{2k_BT}\bigg]= 1+ O(\alpha^2).  
\end{align*}
Therefore, 
\begin{align*}
 \exp\bigg[\!-\!\frac{m_i({\bfv}_i\!-\!\alpha{\bfu}^{\alpha}_i)^2}{2k_BT}\!\bigg]&\!=\!\exp \bigg[-\frac{m_i({\bfv}_i)^2}{2k_BT}\bigg]\bigg(1+\alpha\frac{m_i{\bfu}_i^\alpha\cdot{\bfv}_i}{k_BT}+ O(\alpha^2)\bigg) \big(1+ O(\alpha^2)\big)
 \\
 &=\exp \bigg[-\frac{m_i({\bfv}_i)^2}{2k_BT}\bigg]\bigg(1+\alpha\frac{m_i{\bfu}_i^\alpha\cdot{\bfv}_i}{k_BT}+ O(\alpha^2)\bigg).
\end{align*}
Thus,
\begin{equation*}
f_i^{\alpha}=c_i^{\alpha}\bigg( \frac{m_i}{2\pi k_BT} \bigg)^\frac{3}{2}\exp\bigg[-\frac{m_i({\bfv}_i)^2}{2k_BT}\bigg]\left ( 1+\alpha\frac{m_i{\bfu}_i\cdot{\bfv}_i}{k_BT}+O(\alpha^2) \right ).
\end{equation*}
Similarly,
\begin{equation*}
f_s^{\alpha}=c_s^{\alpha}\bigg( \frac{m_i}{2\pi k_BT} \bigg)^\frac{3}{2}\exp\bigg[-\frac{m_s({\bfv}_s)^2}{2k_BT}\bigg]\left ( 1+\alpha\frac{m_s{\bfu}_s\cdot{\bfv}_s}{k_BT}+O(\alpha^2) \right ).
\end{equation*}
Taking the product of $f_i^\alpha$ and $f_s^\alpha$, we obtain
\begin{multline}
f_i^{\alpha}f_s^{\alpha}=c_i^{\alpha}\bigg( \frac{m_i}{2\pi k_BT} \bigg)^\frac{3}{2}c_s^{\alpha}\bigg( \frac{m_i}{2\pi k_BT} \bigg)^\frac{3}{2}\exp\bigg[-\frac{m_i({\bfv}_i)^2}{2k_BT}\bigg]
\\
\times \exp\bigg[-\frac{m_s({\bfv}_s)^2}{2k_BT}\bigg]\bigg(1+\alpha\frac{m_i{\bfu}_i\cdot{\bfv}_i}{k_BT}+\alpha\frac{m_s{\bfu}_s\cdot{\bfv}_s}{k_BT}+O(\alpha^2)\bigg).
\label{eq:ff_final}
\end{multline}
Substituting \eqref{eq:ff_final} into \eqref{eq:ff}, expanding and using the definitions given in \eqref{eq:definitions} gives the desired result
\qed
\end{proof}
\section{The Maxwell-Stefan Diffusion Limit}
In this section, we obtain the continuity equations, the momentum balance equations for the species and perform a formal asymptotic analysis of the continuity and momentum balance equations toward a diffusion model of MS type.
\subsection{Continuity Equations for the Species}
The continuity equations for the species can formally be derived from the scaled kinetic equations \eqref{eq:scaled_KE}, by integrating over ${\bfv}_i\in\mathbb{R}^3$ as shown in the following lemma.
\begin{lem}
The continuity equations for the species in the non-reactive mixture is given by
\begin{equation}
 \frac{\partial c_i^\alpha }{\partial t} + \frac{\partial }{\partial \bfx}(c_i^\alpha \bfu_i^\alpha)= 0 \quad i=1,\dots,4.
\label{eq:continuity_equation}
\end{equation}
\end{lem}
\begin{proof}
\smartqed
Integrating both sides of  \eqref{eq:scaled_KE} with respect to ${\bfv}_i\in \mathbb{R}^3$, we obtain for $i=1,\dots,4$ 
\begin{equation}
\alpha \frac{\partial }{\partial t} \underbrace{\int\limits_{\mathbb{R}^3} f_i^\alpha d{\bfv}_i}_{c_i^\alpha} + \frac{\partial }{\partial \bfx}\underbrace{\int\limits_{\mathbb{R}^3}
{\bfv}_i f^\alpha_i  d{\bfv}_i}_{\alpha c_i^\alpha \bfu_i^\alpha}= \frac1\alpha \underbrace{\int\limits_{\mathbb{R}^3}  J_i^{\alpha} d{{\bfv}}_i}_{0}+ \frac1\alpha \sum_{\substack{s=1\\ s\neq i}}^{N} \,\,\underbrace{\int\limits_{\mathbb{R}^3}  J_{is}^{\alpha} d{{\bfv}}_i}_{0}. 
\label{eq:mass_balance}
\end{equation}
To show that the first term on the right hand side of \eqref{eq:mass_balance} vanishes, we use the weak form of the mono-species collision operator \eqref{eq:weakform_MSECO} with $\varphi
(\bfv_i)=1$,  see corollary \ref{cor:monospecies_MMKE_con} in subsection 3.2. Similarly, using the weak form of the bi-species collision operator \eqref{eq:weakform_Qisdvi} with $\varphi
(\bfv_i)=1$,  see corollary \ref{cor:Bconservation_MMKE} in subsection 3.2, we obtain that the second term on the right hand side of \eqref{eq:mass_balance} vanishes. Finally, dividing both sides of \eqref{eq:mass_balance} by $\alpha$ gives the desired result.
\qed
\end{proof}
\subsection{Momentum Balance Equations for the Species}
To derive the momentum balance equations for the species from the scaled kinetic equations \eqref{eq:scaled_KE}, we  multiply it by $m_i{\bfv}_i$ and integrating over $\bf{v}_i\in \mathbb{R}^3$  as shown in the following lemma. 
\begin{lem}
The momentum balance equations for the species in the non-reactive mixture is given by
\begin{multline}
\alpha^2 m_i  \frac{\partial }{\partial t} \Big(c_i^\alpha{\bfu}_i^{\alpha} \Big) +{k_B T} {\frac{\partial c_i^\alpha }{\partial \bfx}}  + \alpha^2 m_i  {\frac{\partial }{\partial {\bfx}}} \Big( c_i^\alpha{\bfu}_i^\alpha \otimes {\bfu}_i^\alpha \Big) 
\\[-0.1em]
= \frac{32}{9} \!\sum_{\substack{s=1 \\ s\neq i}}^{N} \!\! \sigma^2_{is}\!\left( 2\pi \mu_{is} k_BT \right)^{\!\frac{1}{2}}\!c_i^\alpha\! c_s^\alpha \!\big( {\bfu}_i^\alpha \!-\! {\bfu}_s^\alpha \big)\!+\! O(\alpha),
\label{eq:Momentum_balance_equations}
\end{multline}
for $i=1,\dots,4$.
\end{lem} 
\begin{proof}
\smartqed
Multiplying both sides of the scaled kinetic equations \eqref{eq:scaled_KE} by $m_i \bfv_i$  and integrating with respect to $\bf{v}_i\in \mathbb{R}^3$, we obtain for $i=1,\dots,4$,
\begin{multline}
\alpha\frac{\partial}{\partial t} \underbrace{ \left(\int_{\mathbb{R}^3}m_i{{\bfv}}_i f_i^\alpha d{{\bfv}}_i \right)}_{\alpha m_i c_i^\alpha \bfu_i^\alpha} +\frac{\partial}{\partial {\bfx}}\underbrace{\left( \int_{\mathbb{R}^3}m_i( {\bfv}_i\otimes {\bfv}_i)f_i^\alpha d{\bfv}_i \right)}_{ c_i^\alpha k_B T+ \alpha^2 m_i\left( c_i^\alpha{\bfu}_i^\alpha \otimes {\bfu}_i^\alpha \right)}
\\[-0.2em]
=\frac{1}{\alpha} \underbrace{\int_{\mathbb{R}^3}m_i {\bfv}_i J_i^{\alpha} d{\bfv}_i}_{0}   
+\frac{1}{\alpha} \underbrace{\sum_{\substack{s=1 \\ s\neq i}}^{N}\int_{\mathbb{R}^3}m_i {\bfv}_i J_{is}^{\alpha} d{\bfv}_i}_{{\cal O}_i}.
\label{eq:momentum_balance1}
\end{multline}
\noindent To show that the first term on the right hand side of \eqref{eq:momentum_balance1} vanishes, we use the weak form \eqref{eq:weakform_MSECO} with $\varphi({\bfv}_i)=m_i{\bfv}_i$ together with the bi-species conservation of momentum given in \eqref{eq:ECOM}, see corollary \ref{cor:monospecies_MMKE_con} in subsection 3.2. 
Concerning the second term on the right hand side of \eqref{eq:momentum_balance1}, using the definition of $J_{is}^{\alpha}$ given in \eqref{eq:scaled_KE}, we obtain
\begin{align*}
{\cal O}_i&=\sum_{\substack{s=1 \\ s\neq i}}^{N}\sigma^2 _{is}\int_{\mathbb{R}^3}\int_{{\mathbb{R}}^3} \int _{{\mathbb{S}^2_+}}m_i {\bfv}_i\left [ {f}_i^{\alpha'}{f^{\alpha '}_s}-f^\alpha_i f^\alpha_s \right] \left \langle \epsilon ,{\bfv}_i-{\bfv}_s \right\rangle d\epsilon\, d{\bfv}_s\, d{\bfv}_i
\\
&=\alpha\sum_{\substack{s=1 \\ s\neq i}}^{N}\sigma^2 _{is}\int_{\mathbb{R}^3}\int_{{\mathbb{R}}^3} \int _{{\mathbb{S}^2_+}}m_i {\bfv}_i M_i^\alpha M_s^\alpha  {\ba}_{is}\cdot({\bfv}_i'-{\bfv}_i)\left \langle \epsilon ,{\bfv}_i-{\bfv}_s \right\rangle d\epsilon\, d{\bfv}_s d{\bfv}_i+O(\alpha^2),
\end{align*}
where we have used \eqref{eq:f_prime_minus_f} to obtain the previous equation. From the definition of post-collisional velocity \eqref{eq:EPCV}, we can write
\begin{align*}
{\bfv}_i'- {\bfv}_i \!&=- \frac{2\mu_{is}}{m_i} \!\left \langle \epsilon , {\bfv}_i \!-\! {\bfv}_s \right \rangle \! \epsilon
\\[-0.05em]
&=\! -\frac{2m_i m_s}{(m_i+m_s)m_i} \!\left \| {\bfv}_i - {\bfv}_s \right \|\cos\theta \frac{(\bfv_i-\bfv_s)}{ \left \| {\bfv}_i - {\bfv}_s \right \|}
\\[-0.05em]
&=\! -\frac{2 m_s}{(m_i+m_s)} \!(\bfv_i-\bfv_s)\cos\theta.
\end{align*}
Setting ${\Vv}=\bfv_i-\bfv_s$, we obtain
\begin{equation}
{\bfv}_i'- {\bfv}_i=-\frac{2 m_s}{m_i+m_s}{\Vv}\cos\theta.
\label{eq:postcollisional_velocity}
\end{equation}
Therefore,
\begin{align*}
{\cal O}_i&\!\!=\!-  \alpha\!\!\sum_{\substack{s=1\\ s\neq i}}^{N}\sigma^2 _{is} 2\mu_{is}\!\! \int_{\mathbb{R}^3}\!\!\int_{{\mathbb{R}}^3} \!\!\int _{{\mathbb{S}^2_+}}\!\!\! {\bfv}_i M_i^\alpha M_s^\alpha \Big( {\ba}_{is}\cdot {\Vv} \Big)\cos\theta \left \langle \epsilon , {\bfv}_i \!-\! {\bfv}_s \right \rangle d\epsilon d{\bfv}_sd{\bfv}_i+O(\alpha^2)
\\[-0.2em]
&\!\!=\!  -\alpha\!\!\sum_{\substack{s=1\\ s\neq i}}^{N}\sigma^2 _{is} 2\mu_{is} \!\! \int_{\mathbb{R}^3}\!\!\int_{{\mathbb{R}}^3} \!\!\int _{{\mathbb{S}^2_+}}\!\! {\bfv}_i M_i^\alpha M_s^\alpha \Big( {\ba}_{is}\cdot {\Vv} \Big)\left \| {\bfv}_i\! -\! {\bfv}_s \right \|\cos^2\theta\,d\epsilon d{\bfv}_sd{\bfv}_i+O(\alpha^2)
\\[-0.2em]
&\!\!=\! - \alpha\!\!\sum_{\substack{s=1\\ s\neq i}}^{N}\sigma^2 _{is} 2\mu_{is} \!\! \int_{\mathbb{R}^3}\!\!\int_{{\mathbb{R}}^3} \!\!\int _0^{\frac{\pi}{2}}\!\!\!\int_0^{2\pi}\!\!\! {\bfv}_i M_i^\alpha M_s^\alpha \!\Big(\! {\ba}_{is}\!\cdot\! {\Vv}\! \Big)V \cos^2\!\theta \sin\theta d\phi\! d\theta \!d{\bfv}_s\,d{\bfv}_i\!+\!O(\alpha^2)
\\[-0.2em]
&\!\!=\!-\alpha\!\!\sum_{\substack{s=1\\ s\neq i}}^{N}\sigma^2 _{is} 2\mu_{is} \!\! \int_{\mathbb{R}^3}\!\!\int_{{\mathbb{R}}^3}\!\! {\bfv}_i M_i^\alpha M_s^\alpha \!\Big(\! {\ba}_{is}\!\cdot\! {\Vv} \!\Big)\!V d{\bfv}_s d{\bfv}_i\!\underbrace{\int_0^{\frac{\pi}{2}}\!\!\!\!\cos^2\theta \sin\theta d\theta}_{\frac{1}{3}} \underbrace{\int_0^{2\pi}\!\!\! d\phi}_{2\pi}\!+\!O(\alpha^2)
\\[-0.2em]
&=-\alpha\sum_{\substack{s=1\\ s\neq i}}^{N}\sigma^2 _{is} \frac{4\pi\mu_{is}}{3} \!\! \int_{\mathbb{R}^3}\!\!\int_{{\mathbb{R}}^3}\!\!\! {\bfv}_i M_i^\alpha M_s^\alpha \Big( {\ba}_{is}\cdot {\Vv} \Big)V\, d{\bfv}_s\,d{\bfv}_i+O(\alpha^2)
\\[-0.2em]
&\!\!=\!-\alpha\!\sum_{\substack{s=1\\ s\neq i}}^{N}\!\sigma^2 _{is}\! \frac{4\pi\mu_{is}}{3} \frac{(m_im_s)^{\!\frac{3}{2}}}{(2\pi k_B T)^3}c_i^\alpha c_s^\alpha\!\! \int_{\mathbb{R}^3}\!\!\int_{{\mathbb{R}}^3}
\!\!\! {\bfv}_i \exp\!\Bigg(\!\!-\!\frac{m_i(\bfv_i)^2\!+\!m_s(\bfv_s)^2}{2k_B T}\!\Bigg)
\\[-0.2em]
&\hspace{5.5cm}\times\Big( {\ba}_{is}\cdot {\Vv} \Big)V d{\bfv}_s d{\bfv}_i +O(\alpha^2).
\end{align*}
 It will be convenient to transform the above  six fold integral from $\bfv_i$ and $\bfv_s$  to the center of mass velocity
\begin{equation}
{\Xx} \!=\!\frac{(m_i {\bfv}_i+m_s {\bfv}_s)}{m_i+m_s} \quad\Longleftrightarrow \quad (m_i+m_s){\Xx}= m_i {\bfv}_i+m_s {\bfv}_s,
\label{eq:COMV}
\end{equation}
and relative velocity
\begin{equation}
{\Vv} \!=\! {\bfv}_i - {\bfv}_s \quad \Longleftrightarrow \quad {\bfv}_i={\Vv} + {\bfv}_s.
\label{eq:RV}
\end{equation}
Observe that by substituting  \eqref{eq:RV} into \eqref{eq:COMV}, we obtain
\begin{align*}
(m_i+m_s){\Xx} &= m_i({\Vv} + {\bfv}_s)+m_s {\bfv}_s
\\
&=m_i{\Vv}+m_i {\bfv}_s + m_s{\bfv}_s
\\
&=m_i{\Vv}+(m_i+m_s)\bfv_s.
\end{align*}
Dividing both sides by $m_i+m_s$, we obtain
\begin{equation}
{\Xx}={\bfv_s}+\frac{m_i{\Vv}}{m_i+m_s}\quad \Longleftrightarrow \quad {\bfv_s}={\Xx}-\frac{m_i{\Vv}}{m_i+m_s}.
\label{eq:v_s}
\end{equation}
Also, substituting \eqref{eq:v_s} into  \eqref{eq:RV}, we obtain
\begin{align}
\nonumber\bfv_i&={\Vv}+{\Xx}-\frac{m_i{\Vv}}{m_i+m_s}
\\
\nonumber&=\frac{m_i{\Vv}+m_s{\Vv}+(m_i+m_s){\Xx}-m_i{\Vv}}{(m_i+m_s)}
\\
\nonumber&=\frac{m_s{\Vv}+(m_i+m_s){\Xx}}{m_i+m_s}
\\
&={\Xx}+\frac{m_s{\Vv}}{m_i+m_s}.
\label{eq:v_i}
\end{align}
Furthermore, using \eqref{eq:v_i}, we obtain
\begin{align*}
m_i({\bfv}_i)^2&=m_i\Bigg({\Xx}+\frac{m_s {\Vv} }{m_i+m_s}\Bigg)^2
\\
 &=m_i\left \{ {\Xx}\left ( {\Xx}+\frac{m_s {\Vv} }{m_i+m_s}\right )+\frac{m_s {\Vv} }{m_i+m_s}\left ( {\Xx}+\frac{m_s {\Vv} }{m_i+m_s} \right ) \right \}
\\
&=m_i\left \{ {{X}}^2+2\frac{m_s}{m_i+m_s}({\Vv}\cdot{\Xx})+\frac{m^2_s}{(m_i+m_s)^2}{{V}}^2\right \},
\end{align*}
where  $X=\left \| {\Xx} \right \|$ and $V=\left \| {\Vv} \right \|$. Similarly, using \eqref{eq:v_s}, we obtain
\begin{equation*}
m_s({\bfv}_s)^2=m_s\left \{ {{X}}^2-2\frac{m_i}{m_i+m_s}({\Vv}\cdot{\Xx}) +\frac{m^2_i}{(m_i+m_s)^2}{{V}}^2\right \}.
\end{equation*}
Thus,
\begin{align}
\nonumber m_i({\bfv}_i)^2+m_s({\bfv}_s)^2&=
(m_i+m_s){{X}}^2+\frac{\mu_{is}m_s}{m_i+m_s}V^2+\frac{\mu_{is}m_i}{m_i+m_s}V^2
\\
\nonumber&=(m_i+m_s){{X}}^2+\frac{\mu_{is}}{m_i+m_s}V^2 (m_i+m_s)
\\
&=(m_i+m_s){{X}}^2+\mu_{is}{{V}}^2. 
\label{eq:X_V}
\end{align}
Therefore, using the fact that the Jacobian of the transformation $(\bfv_i,\bfv_s) \mapsto ({\Vv}, {\Xx})$ has absolute value $1$, we obtain

\begin{align*}
{\cal O}_i&=-\alpha\!\sum_{\substack{s=1\\ s\neq i}}^{N}\!\sigma^2 _{is}\!\frac{4\pi\mu_{is}}{3} \frac{(m_im_s)^{\frac{3}{2}}}{(2\pi k_B T)^3}c_i^\alpha c_s^\alpha \int_{\mathbb{R}^3}\int_{{\mathbb{R}}^3}\Bigg(\!{\Xx}+\frac{m_s{\Vv}}{m_i+m_s}\Bigg) 
\\
&\hspace{3cm}\times \exp\Bigg(\!-\frac{(m_i+m_s){{X}}^2\!+\!\mu_{is}{{V}}^2}{2k_B T}\!\Bigg)( {\ba}_{is}\!\cdot\! {\Vv}) V d{\Xx}d{\Vv}+O(\alpha^2)
\\
&=-\alpha\!\!\sum_{\substack{s=1\\ s\neq i}}^{N}\sigma^2 _{is} \frac{4\pi\mu_{is}}{3} \frac{(m_im_s)^{\!\frac{3}{2}}}{(2\pi k_B T)^3}c_i^\alpha c_s^\alpha \int_{\mathbb{R}^3}\int_{{\mathbb{R}}^3}{\Xx} \exp\Bigg(\!-\!\frac{(m_i+m_s){{X}}^2\!+\!\mu_{is}{{V}}^2}{2k_B T}\!\Bigg)
\\
&\hspace{8cm}\times \Big(\! {\ba}_{is}\!\cdot\! {\Vv} \!\Big)\! V\! d{\Xx}d{\Vv}
\\
&-\alpha\!\sum_{\substack{s=1\\ s\neq i}}^{N}\!\sigma^2 _{is}\! \frac{4\pi\mu_{is}m_s}{3(m_i\!+\!m_s)} \frac{(m_im_s)^{\!\frac{3}{2}}}{(2\pi k_B T)^3}c_i^\alpha c_s^\alpha\! \int_{\mathbb{R}^3}\!\int_{{\mathbb{R}}^3}\! {\Vv}\exp\!\Bigg(\!-\!\frac{(m_i+m_s){{X}}^2\!+\!\mu_{is}{{V}}^2}{2k_B T}\!\Bigg)
\\
&\hspace{7cm}\times\Big(\! {\ba}_{is}\!\cdot\! {\Vv}\! \Big)\!V\! d{\Xx}\!d{\Vv} +O(\alpha^2)
\\
&=-\alpha\sum_{\substack{s=1\\ s\neq i}}^{N}\sigma^2 _{is} \frac{4\pi\mu_{is}}{3} \frac{(m_im_s)^{\!\frac{3}{2}}}{(2\pi k_B T)^3}c_i^\alpha c_s^\alpha\underbrace{\int_{\mathbb{R}^3}
\!\!\!{\Xx} \exp\Bigg(\!-\!\frac{(m_i+m_s){{X}}^2\!}{2k_B T}\!\Bigg)  d{\Xx}}_{\cal A}
\\[-0.1em]
&\hspace{5cm}\times\underbrace{\int_{\mathbb{R}^3}\!\! \exp\Bigg(\!-\!\frac{\mu_{is}{{V}}^2}{2k_B T}\!\Bigg) V\Big( {\ba}_{is}\cdot {\Vv} \Big) \,d{\Vv}}_{\cal B}
\\[-0.1em]
&\!-\!\alpha\!\sum_{\substack{s=1\\ s\neq i}}^{N}\!\sigma^2 _{is} \frac{4\pi\mu_{is}m_s}{3(m_i+m_s)} \frac{(m_im_s)^{\!\frac{3}{2}}}{(2\pi k_B T)^3}c_i^\alpha c_s^\alpha \underbrace{\int_{\mathbb{R}^3}
\!\!\exp\Bigg(\!\!-\!\frac{m_i+m_s{{X}}^2}{2k_B T}\!\Bigg) d{\Xx}}_{\cal C}
\\[-0.1em]
&\hspace{4.2cm}\times\underbrace{\int_{\mathbb{R}^3}
\!\!\! {\Vv}\exp\Bigg(\!\!-\!\frac{\mu_{is}{{V}}^2}{2k_B T}\!\Bigg) V\Big( {\ba}_{is}\!\cdot\! {\Vv} \Big) d{\Vv}}_{\cal D}+ O(\alpha^2).
\end{align*}
Now let us evaluate the integrals ${\cal A}$, ${\cal B}$, ${\cal C}$ and ${\cal D}$.

\noindent Using the fact that any vector ${\Xx}$ can be written in terms of a unit vector as ${\Xx}=X\vec{{\bf{x}}}$, with $X=\left \| {\Xx}\right \|$ and $\vec{{\bf{x}}}$ a unit vector, we can rewrite the integral ${\cal A}$ as given below
\begin{equation}
{\cal A}=\int_{\mathbb{R}^3}\!\! X\vec{{\bf{x}}}\, \exp\Bigg(\!-\!\frac{(m_i+m_s){{X}}^2}{2k_B T}\!\Bigg)  d{\Xx}.
\label{eq:A_first}
\end{equation}
Writing \eqref{eq:A_first} in  spherical coordinates, we obtain
\begin{align*}
{\cal A}&=\int_0^\infty \int_0^\pi \int_0^{2\pi} \!\!X(\hat{\bf{x}}\sin\theta \cos\phi+\hat{\bf{y}}\sin\theta \sin\phi+\hat{\bf{z}}\cos\theta)\, \exp\Bigg(\!-\!\frac{(m_i+m_s){{X}}^2}{2k_B T}\!\Bigg)
\\
&\times X^2 \sin\theta d\phi d\theta dX
\\
&=\int_0^\infty X^3 \exp\Bigg(\!-\!\frac{(m_i+m_s){{X}}^2}{2k_B T}\!\Bigg)dX \Bigg(\hat{\bf{x}}\underbrace{\int_0^\pi \sin^2\theta d\theta}_{\frac{\pi}{2}} \underbrace{\int_0^{2\pi}\cos\phi d\phi}_{0}
\\[-0.2em]
&+\hat{\bf{y}}\underbrace{\int_0^\pi \sin^2\theta d\theta}_{\frac{\pi}{2}} \underbrace{\int_0^{2\pi}\sin\phi}_{0}
+\hat{\bf{z}}\underbrace{\int_0^\pi \cos\theta \sin\theta d\theta}_{0} \underbrace{\int_0^{2\pi} d\phi}_{2\pi} \Bigg),
\end{align*}
where $\hat{\bf{x}}, \hat{\bf{y}}, \hat{\bf{z}}$ are the Cartesian unit vectors in $\mathbb{R}^3$. The integral in $X$ can be evaluated using the integral representation of gamma function defined as given below
\begin{equation}
\int_0^\infty x^n e^{-\eta x^2}dx=\frac{1}{2}\Gamma \bigg( \frac{n+1}{2} \bigg)\bigg( \frac{1}{\eta } \bigg)^{\frac{n+1}{2}}.
\label{eq:int_gamma_function}
\end{equation}
More specifically,
\begin{align*}
\int_0^\infty X^3 \exp\Bigg(\!-\!\frac{(m_i+m_s){{X}}^2}{2k_B T}\!\Bigg)dX&=\frac{1}{2}\Gamma(2) \bigg(\frac{2k_B T}{m_i+m_s}\bigg)^2
\\
&=\frac{1}{2} \bigg(\frac{2k_B T}{m_i+m_s}\bigg)^2.
\end{align*}
Thus,
\begin{equation}
{\cal A}=\frac{1}{2} \bigg(\frac{2k_B T}{(m_i+m_s)}\bigg)^2 \Big[\hat{\bf{x}}\Big(\frac{\pi}{2}\times0\Big)+\hat{\bf{y}}\Big(\frac{\pi}{2}\times 0\Big)+\hat{\bf{z}}(0\times 2\pi)\Big]=0.
\label{eq:A}
\end{equation}
Now, using the fact that any vector ${\Vv}$ can be written in terms of a unit vector as ${\Vv}=V\vec{{\bf{v}}}$, with $V=\left \| {\Vv} \right \|$ and $\vec{{\bf{v}}}$ a unit vector, we obtain that the integral ${\cal B}$ can be rewritten as
\begin{equation*}
{\cal B}=\!{\ba}_{is}\cdot\int_{{\mathbb{R}}^3}\exp\Bigg(\!-\!\frac{\mu_{is}{{V}}^2}{2k_B T}\!\Bigg) V\vec{{\bf{v}}} V\,  d{\Vv}.
\end{equation*}
Transforming to spherical coordinates and integrating, we obtain
\begin{align*}
{\ba}_{is}\cdot\int_{{\mathbb{R}}^3}&\exp\Bigg(\!-\!\frac{\mu_{is}{{V}}^2}{2k_B T}\!\Bigg) V\vec{{\bf{v}}} V\,  d{\Vv}
\\[-0.2em]
&=\!{\ba}_{is}\cdot\!\int_0^\infty\! \int_0^\pi \!\int_0^{2\pi}\!\exp\Bigg(\!-\!\frac{\mu_{is}{{V}}^2}{2k_B T}\!\Bigg) V(\hat{{\bf{x}}}\sin\theta \cos\phi+\hat{{\bf{y}}}\sin\theta\sin\phi+\hat{{\bf{z}}}\cos\theta)
\\[-0.2em]
&\times V^3\sin\theta d\phi d\theta dV
\\[-0.2em]
&=\!{\ba}_{is}\cdot\!\int_0^\infty \!V^4 \exp\Bigg(\!-\!\frac{\mu_{is}{{V}}^2}{2k_B T}\!\Bigg) dV \Bigg(\hat{{\bf{x}}}\!\underbrace{\int_0^\pi \sin^2\theta d\theta}_{\frac{\pi}{2}} \underbrace{\int_0^{2\pi} \cos\phi d\phi}_{0}
\\[-0.2em]
&+\!\hat{{\bf{y}}}\!\underbrace{\int_0^\pi \sin^2\theta d\theta}_{\frac{\pi}{2}} \underbrace{\int_0^{2\pi}\sin\phi d\phi}_{0}+\hat{{\bf{z}}}\underbrace{\int_0^\pi \sin\theta \cos\theta d\theta}_{0} \underbrace{\int_0^{2\pi} d\phi}_{2\pi}\Bigg),
\end{align*}
where $\hat{\bf{x}}, \hat{\bf{y}}, \hat{\bf{z}}$ are the Cartesian unit vectors in $\mathbb{R}^3$. Evaluating the integral with respect to $V$ using \eqref{eq:int_gamma_function}, we obtain
\begin{align*}
 \int_0^\infty \!V^4 \exp\Bigg(\!-\!\frac{\mu_{is}{{V}}^2}{2k_B T}\!\Bigg) dV&=\frac{1}{2}\Gamma\bigg(\frac{5}{2}\bigg) \bigg(\frac{2k_B T}{\mu_{is}}\bigg)^{\frac{5}{2}}
\\[-0.2em]
&=\frac{1}{2} \frac{3}{4}\sqrt{\pi} \bigg(\frac{2k_B T}{\mu_{is}}\bigg)^{\frac{5}{2}}.
\end{align*}
Therefore, 
\begin{equation}
{\cal B}={\ba}_{is}\frac{3\sqrt{\pi}}{8}\bigg(\frac{2k_B T}{\mu_{is}}\bigg)^{\frac{5}{2}} \times \Big[\hat{\bf{x}}\Big(\frac{\pi}{2}\times0\Big)+\hat{\bf{y}}\Big(\frac{\pi}{2}\times 0\Big)+\hat{\bf{z}}(0\times 2\pi)\Big]=0.
\label{eq:zero_moment_v_integral}
\end{equation}
Transforming the integral ${\cal C}$ to spherical coordinates and integrating, we obtain
\begin{align*}
{\cal C}&=\int_0^\infty \int_0^\pi \int_0^{2\pi} \exp\Bigg(\!\!-\!\frac{(m_i+m_s){{X}}^2}{2k_B T}\!\Bigg) X^2 \sin\theta d\phi d\theta dX
\\[-0.1em]
&\!=\!\int_0^\infty X^2 \exp\Bigg(\!\!-\!\frac{(m_i+m_s){{X}}^2}{2k_B T}\!\Bigg)dX \underbrace{\int_0^\pi \sin\theta d\theta}_{2} \underbrace{\int_0^{2\pi}   d\phi}_{2\pi}
\\[-0.1em]
&\!=\!4\pi\int_0^\infty X^2 \exp\Bigg(\!\!-\!\frac{(m_i+m_s){{X}}^2}{2k_B T}\!\Bigg)dX.
\end{align*}
Evaluating the previous integral using the integral representation of gamma \eqref{eq:int_gamma_function}, we obtain
\begin{align}
{\cal C}&=4\pi \frac{1}{2} \Gamma \bigg(\frac{3}{2}\bigg) \bigg(\frac{2k_B T}{m_i+m_s}\bigg)^{\frac{3}{2}}
\\[-0.1em]
\nonumber&=2\pi \frac{\sqrt{\pi}}{2} \bigg(\frac{2k_B T}{m_i+m_s}\bigg)^{\frac{3}{2}}
\\[-0.1em]
&= \bigg(\frac{2\pi k_B T}{m_i+m_s}\bigg)^{\frac{3}{2}}.
\label{eq:zero_moment_x_integral}
\end{align}
Finally, the integral ${\cal D}$ can be rewritten as
\begin{align*}
{\cal D}&=\int_{\mathbb{R}^3}\!\!\! V\vec{{\bf{v}}}\,\exp\Bigg(\!\!-\!\frac{\mu_{is}{{V}}^2}{2k_B T}\!\Bigg) \Big( {\ba}_{is}\!\cdot\! {\Vv} \Big)V d{\Vv}
\\
&=\int_{\mathbb{R}^3}\!\!\! V\vec{{\bf{v}}}\,\exp\Bigg(\!\!-\!\frac{\mu_{is}{{V}}^2}{2k_B T}\!\Bigg) \Big( a_{is}V \cos\theta \Big)V d{\Vv}
\\
&=a_{is}\int_{\mathbb{R}^3}\!\!\! V\vec{{\bf{v}}}\,\exp\Bigg(\!\!-\!\frac{\mu_{is}{{V}}^2}{2k_B T}\!\Bigg)V^2 \cos\theta d{\Vv}.
\end{align*}
Writing the integral ${\cal D}$ in spherical coordinates, we obtain
\begin{align*}
{\cal D}&=a_{is}\int_0^\infty \int_0^\pi \int_0^{2\pi}\!\!\! V(\hat{{\bf{x}}}\sin\theta \cos\phi+\hat{{\bf{y}}}\sin\theta\sin\phi+\hat{{\bf{z}}}\cos\theta)\,\exp\Bigg(\!\!-\!\frac{\mu_{is}{{V}}^2}{2k_B T}\!\Bigg)
\\
&\times V^2 \cos\theta V^2\sin\theta d\phi d\theta dV
\\[-0.25em]
&=\!a_{is}\underbrace{\int_0^\infty\!\! V^5 \exp\Bigg(\!\!-\!\frac{\mu_{is}{{V}}^2}{2k_B T}\!\Bigg) dV}_{{\cal D}_1} \Bigg(\hat{{\bf{x}}}\!\underbrace{\int_0^\pi \sin^2\theta \cos\theta d\theta}_{0} \underbrace{\int_0^{2\pi} \cos\phi d\phi}_{0}
\\[-0.5em]
&+\!\hat{{\bf{y}}}\!\underbrace{\int_0^\pi \sin^2\theta \cos\theta d\theta}_{0} \underbrace{\int_0^{2\pi}\sin\phi d\phi}_{0}+\hat{{\bf{z}}}\underbrace{\int_0^\pi \sin\theta \cos^2\theta d\theta}_{\frac{2}{3}} \underbrace{\int_0^{2\pi} d\phi}_{2\pi}\Bigg),
\end{align*}
where $\hat{\bf{x}}, \hat{\bf{y}}, \hat{\bf{z}}$ are the Cartesian unit vectors in $\mathbb{R}^3$. Evaluating the integral in $V$  using \eqref{eq:int_gamma_function} to obtain
\begin{align}
\nonumber\int_0^\infty\!\! V^5 \exp\Bigg(\!\!-\!\frac{\mu_{is}{{V}}^2}{2k_B T}\!\Bigg) dV&=
\nonumber\frac{1}{2}\Gamma(3) \bigg(\frac{2k_B T}{\mu_{is}}\bigg)^3
\\
\nonumber&=\frac{1}{2} 2! \bigg(\frac{2k_B T}{\mu_{is}}\bigg)^3
\\
&=\bigg(\frac{2k_B T}{\mu_{is}}\bigg)^3.
\label{eq:D_1}
\end{align}
Substituting \eqref{eq:D_1} into the integral ${\cal D}$, we obtain
\begin{equation}
{\cal D}= a_{is}\bigg(\frac{2k_B T}{\mu_{is}}\bigg)^3\bigg[\hat{{\bf{x}}}(0\times 0)+\hat{{\bf{y}}}(0\times 0)+\hat{{\bf{z}}}\Big(\frac{2}{3}\times 2\pi\Big)\bigg]={\ba}_{is}\frac{4\pi}{3}\bigg(\frac{2k_B T}{\mu_{is}}\bigg)^3.
\label{eq:D}
\end{equation}
Since the integrals ${\cal A}$ and ${\cal B}$ both vanish, we have that ${\cal O}_i$ reduces to
\begin{align*}
{\cal O}_i&=\!-\alpha\!\sum_{\substack{s=1\\ s\neq i}}^{N}\!\sigma^2 _{is} \frac{4\pi\mu_{is}m_s}{3(m_i+m_s)} \frac{(m_im_s)^{\!\frac{3}{2}}}{(2\pi k_B T)^3}c_i^\alpha c_s^\alpha \bigg(\frac{2\pi k_B T}{m_i+m_s}\bigg)^{\frac{3}{2}}{\ba}_{is}\frac{4\pi}{3}\bigg(\frac{2k_B T}{\mu_{is}}\bigg)^3+O(\alpha^2)
\\[-0.2em]
&=-\alpha\!\sum_{\substack{s=1\\ s\neq i}}^{N}\!\sigma^2 _{is} \frac{16\mu_{is}^2}{9k_B T} \bigg(\frac{\mu_{is}}{2\pi k_B T}\bigg)^{\frac{3}{2}}c_i^\alpha c_s^\alpha  \left( \frac{2\pi k_BT}{\mu_{is}} \right)^{\!2} \left( \frac{2k_BT}{\mu_{is}} \right)({\bfu}_i^\alpha-{\bfu}_s^\alpha)+O(\alpha^2)
\\[-0.2em]
&=-\alpha\!\sum_{\substack{s=1\\ s\neq i}}^{N}\!\sigma^2 _{is} \frac{32\mu_{is}}{9} c_i^\alpha c_s^\alpha  \left( \frac{2\pi k_BT}{\mu_{is}} \right)^{\frac{1}{2}} ({\bfu}_i^\alpha-{\bfu}_s^\alpha)+O(\alpha^2)
\\[-0.2em]
&=-\alpha\frac{32}{9} \sum_{\substack{s=1\\ s\neq i}}^{N}\!\sigma^2 _{is} \left( 2\pi \mu_{is} k_BT \right)^{\frac{1}{2}} c_i^\alpha c_s^\alpha ({\bfu}_i^\alpha-{\bfu}_s^\alpha)+O(\alpha^2),
\end{align*}
 where we have used the definition of ${\ba}_{is}$ given in \eqref{eq:definitions}. The proof is complete.
 \qed
\end{proof}
\subsection{Asymptotic Analysis}
Here we present the asymptotic analysis of the species continuity equations and the species momentum balance equations towards a diffusion  model where the diffusion process is described by the MS equations.
\begin{thm}
The Maxwellians defined in \eqref{eq:local_max1} are solutions of the initial-boundary value problem \eqref{eq:Scaled_KE} if $(c_i^\alpha, {\bfu}_i^\alpha)$ solves
\begin{equation}
\begin{aligned}
& \frac{\partial c_i^\alpha }{\partial t} +\frac{\partial }{\partial \bfx}(c_i^\alpha \bfu_i^\alpha)= 0, \quad i=1,\dots,4,
\\[-0.2em]
&\alpha^2 m_i  \frac{\partial }{\partial t} \Big(c_i^\alpha{\bfu}_i^{\alpha} \Big) +{k_B T} {\frac{\partial c_i^\alpha }{\partial \bfx}}  + \alpha^2 m_i  {\frac{\partial }{\partial {\bfx}}} \Big( c_i^\alpha{\bfu}_i^\alpha \otimes {\bfu}_i^\alpha \Big)=\frac{1}{\alpha}{\cal O}_i,
\end{aligned}
\label{eq:Continuity_and_momentum}
\end{equation}
where
\begin{equation*}
\frac{1}{\alpha}{\cal O}_i=-\frac{32}{9} \!\sum_{\substack{s=1 \\ s\neq i}}^{N} \!\! \sigma^2_{is}\!\left( 2\pi \mu_{is} k_BT \right)^{\!\frac{1}{2}}\!c_i^\alpha\! c_s^\alpha \!\big( {\bfu}_i^\alpha \!-\! {\bfu}_s^\alpha \big)+ O(\alpha).
\end{equation*}
Moreover, in the limit as $\alpha \rightarrow 0$ equations \eqref{eq:Continuity_and_momentum} reduces to
\begin{equation}
\begin{aligned}
& \frac{\partial c_i }{\partial t} + \frac{\partial {\bf{J}}_i }{\partial \bfx}= 0, \quad i=1,\dots,4,
\\[-0.2em]
& {\frac{\partial c_i }{\partial \bfx}}=\frac{32}{{9}}\sum_{\substack{s=1\\ s\neq i}}^{N}\sigma^2_{is}
         \left( \frac{2\pi\mu_{is}}{ k_BT} \right)^{\!\!\!\frac{1}{2}} 
         \Big( c_i{\bf{J}}_s - c_s{\bf{J}}_i \Big).
\end{aligned}
\label{eq:MM_equations}
\end{equation}
\end{thm}
\begin{proof}
\smartqed
Putting together equations \eqref{eq:continuity_equation} and \eqref{eq:Momentum_balance_equations} gives \eqref{eq:Continuity_and_momentum}. To obtain \eqref{eq:MM_equations}, we first set  $c_i^\alpha{\bf{u}}_i^\alpha = {\bf{J}}_i^\alpha + c_i^\alpha {\bf{u}}^\alpha$ in \eqref{eq:Continuity_and_momentum}, where ${\bf{u}}^\alpha$ represents the average velocity of the mixture, next we divide both sides of the second equation of \eqref{eq:Continuity_and_momentum} by $k_B T$, then take the limit for $\alpha\rightarrow 0$, where
\begin{equation*}
c_i = \lim_{\alpha\rightarrow 0}c_i^\alpha,    \qquad 
{\bf{J}}_i = \lim_{\alpha\rightarrow 0}{\bf{J}}_i^\alpha,  \qquad  
{\bf{u}} = \lim_{\alpha\rightarrow 0}{\bf{u}}^\alpha, 
\end{equation*}
 and neglect the convective term, that is $\displaystyle\frac{\partial}{\partial \bf x}(c_i {\bf u})=0$.
\qed
\end{proof}
Observe that summing over all species both equations given in  \eqref{eq:MM_equations}, we obtain that $\displaystyle\frac{\partial c }{\partial t}=0$ and $\displaystyle{\frac{\partial c }{\partial \bfx}}=0$, respectively where $\displaystyle c=\sum_{i=1}^{N} c_i$. Therefore, we must have that $c$ is uniform in space and constant in time.
Hence, the second equation 
of system \eqref{eq:MM_equations} 
can be rewritten as       
\begin{equation}
\frac{\partial  c_i}{\partial\bf x} = \frac{1}{c}\sum_{\substack{s=1\\ s\neq i}}^{N}\frac{ c_i{\bf{J}}_s - c_s{\bf{J}}_i }{D_{is}} \;,\quad i=1,\dots,4,
\label{eq:newnew}
\end{equation}
where      
\begin{equation}
D_{is} = \frac{9}{32} \left( \frac{k_BT}{2\pi\mu_{is}} \right)^{\!\!\frac{1}{2}}\frac{1}{c \, \sigma_{is}^2}. 
\label{eq:diffcoeff}
\end{equation}
 The equations given in \eqref{eq:newnew} are the MS equations and the  MS diffusion coefficients are given in \eqref{eq:diffcoeff}. Replacing the second equation of \eqref{eq:MM_equations} with \eqref{eq:newnew}, we obtain 
 \begin{equation}
\begin{aligned}
& \frac{\partial c_i }{\partial t} + \frac{\partial {\bf{J}}_i }{\partial \bfx}= 0 \quad i=1,\dots,4
\\
&\frac{\partial  c_i}{\partial\bf x} = \frac{1}{c}\sum_{\substack{s=1\\ s\neq i}}^{N}\frac{ c_i{\bf{J}}_s - c_s{\bf{J}}_i }{D_{is}}.
\end{aligned}
\label{eq:diffusion_model}
\end{equation}
System \eqref{eq:diffusion_model} above is a diffusion model where the diffusion process is described by the MS equations.
\begin{rem}
System \eqref{eq:diffusion_model} is a diffusion model where the diffusion process is governed by the MS equations. We note that this diffusion model is similar to the ones obtained in \cite{BGP-NA-17,bou-gre-sal-15, HS-MMAS-17} however, the MS diffusion coefficients are different. This difference is as a result of the cross-sections considered. Furthermore,  the MS diffusion coefficients obtained in this work are the same as those obtained in \cite{AGS-19} because the cross-sections considered in both cases are of hard-sphere type, and also because of the chemical regime considered in \cite{AGS-19}, which results in the fact that diffusion effects are not seen on the  chemical reactive terms.
\end{rem}
\section{Conclusion}
In this work, we have studied a  kinetic model for non-reactive mixtures with hard-sphere cross-sections under isothermal condition and obtained a diffusion model where the diffusion process is described by the Maxwell-Stefan equations. More precisely, the diffusion model of MS type was obtained as a hydrodynamic limit of the scaled kinetic equations. In particular, from the species continuity equations and the species momentum balance equations in the limit as the scaling parameter tends to zero.

 The diffusion model of MS type obtained in this work is similar to those obtained in \cite{BGP-NA-17, bou-gre-sal-15, HS-MMAS-17}. However, the MS diffusion coefficients are different. This is because the cross-sections of the kinetic models are different. More precisely, \cite{BGP-NA-17, bou-gre-sal-15, HS-MMAS-17} studied  kinetic models for non-reactive mixtures with Maxwellian molecules, general and analytic cross-sections, respectively, while in this work a kinetic model for non-reactive mixtures with hard-sphere cross-sections was considered.
\begin{acknowledgement}
I thank FCT/Portugal for support through the PhD grant PD/BD/128188/2016 and  Centro de Matem\'{a}tica da Universidade do Minho.
 \end{acknowledgement}

\end{document}